\documentclass[11pt,a4paper,reqno]{amsart}
\usepackage{amsmath}%
\usepackage{amsfonts}%
\usepackage{amssymb}%
\usepackage{graphicx}
\usepackage{hyperref}
\theoremstyle{plain}

\newtheorem{corollary}{Corollary}

\newtheorem{definition}{Definition}
\newtheorem{example}{Example}

\newtheorem{lemma}{Lemma}

\newtheorem{proposition}{Proposition}

\numberwithin{equation}{section}

\newcommand\N{\mathcal{N}}
\newcommand\BH{\mathcal{B}(\cH)}
\newcommand\cM{\mathcal{M}}
\newcommand\PN{\mathcal{P}(\N)}

\newcommand\hA{{\hat A}}

\newcommand\bbR{\mathbb{R}}
\newcommand\bbC{\mathbb{C}}
\newcommand\ra{\rightarrow}

\newcommand\lra{\longrightarrow}
\newcommand\lmt{\longmapsto}
\newcommand\bmeet{\bigwedge}
\newcommand\bjoin{\bigvee}

\newcommand\hP{{\hat P}}
\newcommand\hQ{\hat Q}
\newcommand\cH{\mathcal{H}}

\newcommand\mc[1]{\mathcal{#1}}

\newcommand\deo{\delta^o}

\newcommand\on[1]{\operatorname{#1}}
\newcommand\ps[1]{\underline{#1}}

\newcommand\Sub{\on{Sub}}
\newcommand\Subcl{\on{Sub}_{\on{cl}}}
\newcommand\Sig{\ps{\Sigma}}
\newcommand\VN{\mc{V}(\N)}
\newcommand\meet{\wedge}
\newcommand\join{\vee}
\newcommand\Set{\mathbf{Set}}

\newcommand\PV{\mc{P}(V)}

\newcommand\Ga{\Gamma}
\newcommand\De{\Delta}
\newcommand\Ain[1]{A\,\varepsilon\,#1}
\newcommand\true{\on{true}}
\newcommand\false{\on{false}}
\newcommand\ld{\lambda}
\newcommand\Cp{\mc{C}p}
\newcommand\SetC[1]{\Set^{{#1}^{\rm op}}}
\newcommand\SetVNop{\SetC{\mathcal{V(N)}}}
\newcommand\op{\operatorname{op}}

\begin{document}
\title[Topos-Based Logic for Quantum Systems and Bi-Heyting Algebras]{Topos-Based Logic for Quantum Systems\\and Bi-Heyting Algebras}

\author{Andreas D\"oring}
\address{Andreas D\"oring\newline
\indent Clarendon Laboratory\newline
\indent Department of Physics\newline%
\indent University of Oxford\newline
\indent Parks Road\newline
\indent Oxford OX1 3PU, UK}%
\email{doering@atm.ox.ac.uk}%
%\urladdr{http://www.authortwo.uni-atwo.hu}
%\author{Author Three}
%\address[A. Three]{Author Three address, line 1\newline%
%\indent Author Three address, line 2}
%\email[A.~Three]{author-three@authorthree-inst.edu}%
%\urladdr{http://www.authorthree.uni-athree.edu}
%\thanks{Thanks for Author One.}
%\thanks{Thanks for Author Two.}
%\thanks{This paper is in final form and no version of it will be submitted for
%publication elsewhere.}
\date{December 5, 2013}
%\dedicatory{Dedicated to Professor XY on the occasion of his seventieth birthday.}

\begin{abstract}
To each quantum system, described by a von Neumann algebra of physical quantities, we associate a complete bi-Heyting algebra. The elements of this algebra represent contextualised propositions about the values of the physical quantities of the quantum system.
\end{abstract}
\maketitle

%-----------------------------------------------
\section{Introduction}			\label{Sec_Introd}
Quantum logic started with Birkhoff and von Neumann's seminal article \cite{BvN36}. Since then, non-distributive lattices with an orthocomplement (and generalisations thereof) have been used as representatives of the algebra of propositions about the quantum system at hand. There are a number of well-known conceptual and interpretational problems with this kind of `logic'. For review of standard quantum logic(s), see the article \cite{DCG02}.

In the last few years, a different form of logic for quantum systems based on generalised spaces in the form of presheaves and topos theory has been developed by Chris Isham and this author \cite{DI(1),DI(2),DI(3),DI(4),Doe07b,DI(Coecke)08,Doe09,Doe11}. This new form of logic for quantum systems is based on a certain Heyting algebra $\Subcl{\Sig}$ of clopen, i.e., closed and open subobjects of the spectral presheaf $\Sig$. This generalised space takes the r\^ole of a state space for the quantum system. (All technical notions are defined in the main text.) In this way, one obtains a well-behaved intuitionistic form of logic for quantum systems which moreover has a topological underpinning.

In this article, we will continue the development of the topos-based form of logic for quantum systems. The main new observation is that the complete Heyting algebra $\Subcl{\Sig}$ of clopen subobjects representing propositions is also a complete co-Heyting algebra. Hence, we relate quantum systems to complete bi-Heyting algebras in a systematic way. This includes two notions of implication and two kinds of negation, as discussed in the following sections. 

The plan of the paper is as follows: in section \ref{Sec_Background}, we briefly give some background on standard quantum logic and the main ideas behind the new topos-based form of logic for quantum systems. Section \ref{Sec_BiHeytAlgs} recalls the definitions and main properties of Heyting, co-Heyting and bi-Heyting algebras, section \ref{Sec_SigEtc} introduces the spectral presheaf $\Sig$ and the algebra $\Subcl{\Sig}$ of its clopen subobjects. In section \ref{Sec_RepOfPropsAndBiHeyting}, the link between standard quantum logic and the topos-based form of quantum logic is established and it is shown that $\Subcl{\Sig}$ is a complete bi-Heyting algebra. In section \ref{Sec_NegsAndRegs}, the two kinds of negations associated with the Heyting resp. co-Heyting structure are considered. Heyting-regular and co-Heyting regular elements are characterised and a tentative physical interpretation of the two kinds of negation is given. Section \ref{Sec_Conclusion} concludes.

Throughout, we assume some familiarity with the most basic aspects of the theory of von Neumann algebras and with basics of category and topos theory. The text is interspersed with some physical interpretations of the mathematical constructions.

%-----------------------------------------------
\section{Background}			\label{Sec_Background}
\textbf{Von Neumann algebras.} In this article, we will discuss structures associated with von Neumann algebras, see e.g. \cite{KR83}. This class of algebras is general enough to describe a large variety of quantum mechanical systems, including systems with symmetries and/or superselection rules. The fact that each von Neumann algebra has `sufficiently many' projections makes it attractive for quantum logic. More specifically, each von Neumann algebra is generated by its projections, and the spectral theorem holds in a von Neumann algebra, providing the link between self-adjoint operators (representing physical quantities) and projections (representing propositions).

The reader not familiar with von Neumann algebras can always take the algebra $\BH$ of all bounded operators on a separable, complex Hilbert space $\cH$ as an example of a von Neumann algebra. If the Hilbert space $\cH$ is finite-dimensional, $\dim\cH=n$, then $\BH$ is nothing but the algebra of complex $n\times n$-matrices.

\textbf{Standard quantum logic.} From the perspective of quantum logic, the key thing is that the projection operators in a von Neumann algebra $\N$ form a complete orthomodular lattice $\PN$. Starting from Birkhoff and von Neumann \cite{BvN36}, such lattices (and various kinds of generalisations, which we don't consider here) have been considered as quantum logics, or more precisely as algebras representing propositions about quantum systems.

The kind of propositions that we are concerned with (at least in standard quantum logic) are of the form ``the physical quantity $A$ has a value in the Borel set $\De$ of real numbers'', which is written shortly as ``$\Ain\De$''. These propositions are pre-mathematical entities that refer to the `world out there'. In standard quantum logic, propositions of the form ``$\Ain\De$'' are represented by projection operators via the spectral theorem. If, as we always assume, the physical quantity $A$ is described by a self-adjoint operator $\hA$ in a given von Neumann algebra $\N$, or is affiliated with $\N$ in the case that $\hA$ is unbounded, then the projection corresponding to ``$\Ain\De$'' lies in $\PN$. (For details on the spectral theorem see any book on functional analysis, e.g. \cite{KR83}.)

Following Birkhoff and von Neumann, one then interprets the lattice operations $\meet,\join$ in the projection lattice $\PN$ as logical connectives between the propositions represented by the projections. In this way, the meet $\meet$ becomes a conjunction and the join $\join$ a disjunction. Moreover, the orthogonal complement of a projection, $\hP':=\hat 1-\hP$, is interpreted as negation. Crucially, meets and joins do not distribute over each other. In fact, $\PN$ is a distributive lattice if and only if $\N$ is abelian if and only if all physical quantities considered are mutually compatible, i.e., co-measurable. 

Quantum systems always have some incompatible physical quantities, so $\N$ is never abelian and $\PN$ is non-distributive. This makes the interpretation of $\PN$ as an algebra of propositions somewhat dubious. There are many other conceptual difficulties with quantum logics based on orthomodular lattices, see e.g. \cite{DCG02}. 

\textbf{Contexts and coarse-graining.} The topos-based form of quantum logic that was established in \cite{DI(2)} and developed further in \cite{Doe07b,DI(Coecke)08,Doe09,Doe11} is fundamentally different from standard quantum logic. For some conceptual discussion, see in particular \cite{Doe09}. Two key ideas are \emph{contextuality} and \emph{coarse-graining} of propositions. Contextuality has of course been considered widely in foundations of quantum theory, in particular since Kochen and Specker's seminal paper \cite{KS67}. Yet, the systematic implementation of contextuality in the language of presheaves is comparatively new. It first showed up in work by Chris Isham and Jeremy Butterfield \cite{Ish97,IB98,IB99,IB00,IB00b,IB02,Ish05} and was substantially developed by this author and Isham. For recent, related work see also \cite{HLS09,CHLS09,HLS09b,HLS11} and \cite{AB11,AMS11}.

Physically, a context is nothing but a set of compatible, i.e., co-measurable physical quantities $(A_i)_{i\in I}$. Such a set determines and is determined by an abelian von Neumann subalgebra $V$ of the non-abelian von Neumann algebra $\N$ of (all) physical quantities. Each physical quantity $A_i$ in the set is represented by some self-adjoint operator $\hA$ in $V$.\footnote{From here on, we assume that all the physical quantities $A_i$ correspond to \emph{bounded} self-adjoint operators that lie in $\N$. Unbounded self-adjoint operators affiliated with $\N$ can be treated in a straightforward manner.} In fact, $V$ is generated by the operators $(\hA_i)_{i\in I}$ and the identity $\hat 1$, in the sense that $V=\{\hat 1, \hA_i \mid i\in I\}''$, where $\{S\}''$ denotes the double commutant of a set $S$ of operators (see e.g. \cite{KR83}).\footnote{We will often use the notation $V'$ for a subalgebra of $V$, which does \emph{not} mean the commutant of $V$. We trust that this will not lead to confusion.} Each abelian von Neumann subalgebra $V$ of $\N$ will be called a context, thus identifying the mathematical notion and its physical interpretation. The set of all contexts will be denoted $\VN$. Each context provides one of many `classical perspectives' on a quantum system. We partially order the set of contexts $\VN$ by inclusion. A smaller context $V'\subset V$ represents a `poorer', more limited classical perspective containing fewer physical quantities than $V$.

Each context $V\in\VN$ has a complete Boolean algebra $\PV$ of projections, and $\PV$ clearly is a sublattice of $\PN$. Propositions ``$\Ain\De$'' about the values of physical quantities $A$ in a (physical) context correspond to projections in the (mathematical) context $V$. Since $\PV$ is a Boolean algebra, there are Boolean algebra homomorphisms $\ld:\PV\ra\{0,1\}\simeq\{\false,\true\}$, which can be seen as truth-value assignments as usual. Hence, there are consistent truth-value assignments for all propositions ``$\Ain\De$'' for propositions about physical quantities \emph{within} a context.

The key result by Kochen and Specker \cite{KS67} shows that for $\N=\BH$, $\dim\cH\geq 3$, there are no truth-value assignments for \emph{all} contexts simultaneously in the following sense: there is no family of Boolean algebra homomorphisms $(\ld_V:\PV\ra\{0,1\})_{V\in\VN}$ such that if $V'=V\cap\tilde V$ is a subcontext of both $V$ and $\tilde V$, then $\ld_{V'}=\ld_V|_{V'}=\ld_{\tilde V}|_{V'}$, where $\ld_V|_{V'}$ is the restriction of $\ld_V$ to the subcontext $V'$, and analogously $\ld_{\tilde V}|_{V'}$. As Isham and Butterfield realised \cite{IB98,IB00}, this means that a certain presheaf has no global elements. In \cite{Doe05}, it is shown that this result generalises to all von Neumann algebras without a type $I_2$-summand.

In the topos approach to quantum theory, propositions are represented not by projections, but by suitable subobjects of a quantum state space. An obstacle arises since the Kochen-Specker theorem seems to show that such a quantum state space cannot exist. Yet, if one considers presheaves instead of sets, this problem can be overcome. The presheaves we consider are `varying sets' $(\ps S_V)_{V\in\VN}$, indexed by contexts. Whenever $V'\subset V$, there is a function defined from $\ps S_V$, the set associated with the context $V$, to $\ps S_{V'}$, the set associated with the smaller context $V'$. This makes $\ps S=(\ps S_V)_{V\in\VN}$ into a contravariant, $\Set$-valued functor.

Since by contravariance we go from $\ps S_V$ to $\ps S_{V'}$, there is a built-in idea of \emph{coarse-graining}: $V$ is the bigger context, containing more self-adjoint operators and more projections than the smaller context $V'$, so we can describe more physics from the perspective of $V$ than from $V'$. Typically, the presheaves defined over contexts will mirror this fact: the component $\ps S_V$ at $V$ contains more information (in a suitable sense, to be made precise in the examples in section \ref{Sec_SigEtc}) than $\ps S_{V'}$, the component at $V'$. Hence, the presheaf map $\ps S(i_{V'V}):\ps S_V\ra\ps S_{V'}$ will implement a form of coarse-graining of the information available at $V$ to that available at $V'$.

The subobjects of the quantum state space, which will be called the spectral presheaf $\Sig$, form a (complete) Heyting algebra. This is typical, since the subobjects of any object in a topos form a Heyting algebra. Heyting algebras are the algebraic representatives of (propositional) intuitionistic logics. In fact, we will not consider \emph{all} subobjects of the spectral presheaf, but rather the so-called clopen subobjects. The latter also form a complete Heyting algebra, as was first shown in \cite{DI(2)} and is proven here in a different way, using Galois connections, in section \ref{Sec_RepOfPropsAndBiHeyting}. The difference between the set of all subobjects of the spectral presheaf and the set of clopen subobjects is analogous to the difference between all subsets of a classical state space and (equivalence classes modulo null subsets of) measurable subsets.

Together with the representation of states (which we will not discuss here, but see \cite{DI(2),Doe08,Doe09,DI12}), these constructions provide an intuitionistic form of logic for quantum systems. Moreover, there is a clear topological underpinning, since the quantum state space $\Sig$ is a generalised space associated with the nonabelian algebra $\N$.

The construction of the presheaf $\Sig$ and its algebra of subobjects incorporates the concepts of contextuality and coarse-graining in a direct way, see sections \ref{Sec_SigEtc} and \ref{Sec_RepOfPropsAndBiHeyting}.

%----------------------------------------------------------
\section{Bi-Heyting algebras}			\label{Sec_BiHeytAlgs}
The use of bi-Heyting algebras in superintuitionistic logic was developed by Rauszer \cite{Rau73,Rau77}. Lawvere emphasised the importance of co-Heyting and bi-Heyting algebras in category and topos theory, in particular in connection with continuum physics \cite{Law86,Law91}. Reyes, with Makkai \cite{MR95} and Zolfaghari \cite{RZ96}, connected bi-Heyting algebras with modal logic. In a recent paper, Bezhanishvili et al. \cite{BBGK10} prove (among other things) new duality theorems for bi-Heyting algebras based on bitopological spaces. Majid has suggested to use Heyting and co-Heyting algebras within a tentative representation-theoretic approach to the formulation of quantum gravity \cite{Maj95,Maj08}.

As far as we are aware, nobody has connected quantum systems and their logic with bi-Heyting algebras before.

The following definitions are standard and can be found in various places in the literature; see e.g. \cite{RZ96}.

A \textbf{Heyting algebra $H$} is a lattice with bottom element $0$ and top element $1$ which is a cartesian closed category. In other words, $H$ is a lattice such that for any two elements $A,B\in H$, there exists an exponential $A\Rightarrow B$, called the \textbf{Heyting implication (from $A$ to $B$)}, which is characterised by the adjunction
\begin{equation}
			C\leq(A\Rightarrow B)\quad\text{if and only if}\quad C\wedge A\leq B.
\end{equation}
This means that the product (meet) functor $A\meet_:H\ra H$ has a right adjoint $A\Rightarrow\_:H\ra H$ for all $A\in H$.

It is straightforward to show that the underlying lattice of a Heyting algebra is distributive. If the underlying lattice is complete, then the adjoint functor theorem for posets shows that for all $A\in H$ and all families $(A_i)_{i\in I}\subseteq H$, the following infinite distributivity law holds:
\begin{equation}
			A\meet\bjoin_{i\in I}A_i = \bjoin_{i\in I}(A\meet A_i).
\end{equation}

The \textbf{Heyting negation} is defined as
\begin{align}
			\neg: H &\lra H^{\op}\\			\nonumber
			A &\lmt (A\Rightarrow 0).
\end{align}
The defining adjunction shows that $\neg A=\bjoin\{B\in H \mid A\meet B=0\}$, i.e., $\neg A$ is the largest element in $H$ such that $A\meet\neg A=0$. Some standard properties of the Heyting negation are:
\begin{align}
			& A\leq B\text{ implies }\neg A\geq\neg B,\\
			& \neg\neg A\geq A,\\
			& \neg\neg\neg A=\neg A\\
			& \neg A\join A\leq 1.
\end{align}
Interpreted in logical terms, the last property on this list means that in a Heyting algebra the law of excluded middle need not hold: in general, the disjunction between a proposition represented by $A\in H$ and its Heyting negation (also called Heyting complement, or pseudo-complement) $\neg A$ can be smaller than $1$, which represents the trivially true proposition. Heyting algebras are algebraic representatives of (propositional) intuitionistic logics.

A canonical example of a Heyting algebra is the topology $\mc T$ of a topological space $(X,\mc T)$, with unions of open sets as joins and intersections as meets.

A \textbf{co-Heyting algebra} (also called \textbf{Brouwer algebra $J$}) is a lattice with bottom element $0$ and top element $1$ such that the coproduct (join) functor $A\join\_:J\ra J$ has a left adjoint $A\Leftarrow\_:J\ra J$, called the \textbf{co-Heyting implication (from $A$)}. It is characterised by the adjunction
\begin{equation}
			(A\Leftarrow B)\leq C\text{ iff }A\leq B\join C.
\end{equation}

It is straightforward to show that the underlying lattice of a co-Heyting algebra is distributive. If the underlying lattice is complete, then the adjoint functor theorem for posets shows that for all $A\in J$ and all families $(A_i)_{i\in I}\subseteq J$, the following infinite distributivity law holds:
\begin{equation}
			A\join\bmeet_{i\in I}A_i = \bmeet_{i\in I}(A\join A_i).
\end{equation}

The \textbf{co-Heyting negation} is defined as
\begin{align}
			\sim: J &\lra J^{\op}\\
			A &\lmt (1\Leftarrow A).
\end{align}
The defining adjunction shows that $\sim A=\bmeet\{B\in J \mid A\join B=1\}$, i.e., $\sim A$ is the smallest element in $J$ such that $A\join \sim A=1$. Some properties of the co-Heyting negation are:
\begin{align}
			& A\leq B\text { implies }\sim A\geq\sim B,\\
			& \sim\sim A\leq A,\\
			&\sim\sim\sim A=\sim A\\
			&\sim A\meet A\geq 0.
\end{align}
Interpreted in logical terms, the last property on this list means that in a co-Heyting algebra the law of noncontradiction does not hold: in general, the conjunction between a proposition represented by $A\in J$ and its co-Heyting negation $\sim A$ can be larger than $0$, which represents the trivially false proposition. Co-Heyting algebras are algebraic representatives of (propositional) paraconsistent logics.

We will not discuss paraconsistent logic in general, but in the final section \ref{Sec_NegsAndRegs}, we will give and interpretation of the co-Heyting negation showing up in the form of quantum logic to be presented in this article.

A canonical example of a co-Heyting algebra is given by the closed sets $\mc C$ of a topological space, with unions of closed sets as joins and intersections as meets.

Of course, Heyting algebras and co-Heyting algebras are dual notions. The opposite $H^{\op}$ of a Heyting algebra is a co-Heyting algebra and vice versa.

A \textbf{bi-Heyting algebra $K$} is a lattice which is a Heyting algebra and a co-Heyting algebra. For each $A\in K$, the functor $A\meet\_:K \ra K$ has a right adjoint $A\Rightarrow\_:K\ra K$, and the functor $A\join\_:K\ra K$ has a left adjoint $K\Leftarrow\_:K\ra K$. A bi-Heyting algebra $K$ is called complete if it is complete as a Heyting algebra and complete as a co-Heyting algebra.

A canonical example of a bi-Heyting algebra is a Boolean algebra $\mc B$. (Note that by Stone's representation theorem, each Boolean algebra is isomorphic to the algebra of clopen, i.e., closed and open, subsets of its Stone space. This gives the connection with the topological examples.) In a Boolean algebra, we have for the Heyting negation that, for all $A\in\mc B$,
\begin{equation}
			A\join\neg A=1,
\end{equation}
which is the characterising property of the co-Heyting negation. In fact, in a Boolean algebra, $\neg =\sim$.

%----------------------------------------------------------
\section{The spectral presheaf of a von Neumann algebra and clopen subobjects}			\label{Sec_SigEtc}
With each von Neumann algebra $\N$, we associate a particular presheaf, the so-called spectral presheaf. A distinguished family of subobjects, the so-called clopen subobjects, are defined and their interpretation is given: clopen subobjects can be seen as families of local propositions, compatible with respect to coarse-graining. The constructions presented here summarise those discussed in \cite{DI(1),DI(2),DI(Coecke)08,DI12}.

Let $\N$ be a von Neumann algebra, and let $\VN$ be the set of its abelian von Neumann subalgebras, partially ordered under inclusion. We only consider subalgebras $V\subset\N$ which have the same unit element as $\N$, given by the identity operator $\hat 1$ on the Hilbert space on which $\N$ is represented. By convention, we exclude the trivial subalgebra $V_0=\bbC\hat 1$ from $\VN$. (This will play an important r\^ole in the discussion of the Heyting negation in section \ref{Sec_NegsAndRegs}.) The poset $\VN$ is called the \textbf{context category of the von Neumann algebra $\N$}.

For $V',V\in\VN$ such that $V'\subset V$, the inclusion $i_{V'V}:V'\hookrightarrow V$ restricts to a morphism $i_{V'V}|_{\mc P(V')}:\mc P(V')\ra\PV$ of complete Boolean algebras. In particular, $i_{V'V}$ preserves all meets, hence it has a left adjoint
\begin{align}
			\deo_{V,V'}: \PV &\lra \mc P(V')\\			\nonumber
			\hP &\lmt \deo_{V,V'}(\hP)=\bmeet\{\hQ\in V' \mid \hQ\geq \hP\}
\end{align}
that preserves all joins, i.e., for all families $(\hP_i)_{i\in I}\subseteq\PV$, it holds that
\begin{equation}
			\deo_{V,V'}(\bjoin_{i\in I}\hP_i)=\bjoin_{i\in I}\deo_{V,V'}(\hP_i),
\end{equation}
where the join on the left hand side is taken in $\PV$ and the join on the right hand side is in $\mc P(V')$. If $W\subset V'\subset V$, then $\deo_{V,W}=\deo_{V',W}\circ\deo_{V,V'}$, obviously.

We note that distributivity of the lattices $\PV$ and $\mc P(V')$ plays no r\^ole here. If $\N$ is a von Neumann algebra and $\cM$ is any von Neumann subalgebra such that their unit elements coincide, $\hat 1_{\cM}=\hat 1_{\N}$, then there is a join-preserving map
\begin{align}
			\deo_{\N,\cM'}: \PN &\lra \mc P(\cM)\\			\nonumber
			\hP &\lmt \deo_{\N,\cM}(\hP)=\bmeet\{\hQ\in \mc P(\cM) \mid \hQ\geq \hP\}.
\end{align}

Recall that the Gel'fand spectrum $\Sigma(A)$ of an abelian $C^*$-algebra $A$ is the set of algebra homomorphisms $\ld:A\ra\bbC$. Equivalently, the elements of the Gel'fand spectrum $\Sigma(A)$ are the pure states of $A$. The set $\Sigma(A)$ is given the relative weak*-topology (as a subset of the dual space of $A$), which makes it into a compact Hausdorff space. By Gel'fand-Naimark duality, $A\simeq C(\Sigma(A))$, that is, $A$ is isometrically $*$-isomorphic to the abelian $C^*$-algebra $C(\Sigma(A))$ of continuous, complex-valued functions on $\Sigma(A)$, equipped with the supremum norm. If $A$ is an abelian von Neumann algebra, then $\Sigma(A)$ is extremely disconnected.

We now define the main object of interest:
\begin{definition}
Let $\N$ be a von Neumann algebra. The \textbf{spectral presheaf $\Sig$ of $\N$} is the presheaf over $\VN$ given
\begin{itemize}
	\item [(a)] on objects: for all $V\in\VN$, $\Sig_V:=\Sigma(V)$, the Gel'fand spectrum of $V$,
	\item [(b)] on arrows: for all inclusions $i_{V'V}:V'\hookrightarrow V$,
	\begin{align}
				\Sig(i_{V'V}): \Sig_V &\lra \Sig_{V'}\\			\nonumber
				\ld &\lmt \ld|_{V'}.
	\end{align}
\end{itemize}
\end{definition}
The restriction maps $\Sig(i_{V'V})$ are well-known to be continuous, surjective maps with respect to the Gel'fand topologies on $\Sig_V$ and $\Sig_{V'}$, respectively. They are also open and closed, see e.g. \cite{DI(2)}.

We equip the spectral presheaf with a distinguished family of subobjects (which are subpresheaves):
\begin{definition}
A subobject $\ps S$ of $\Sig$ is called \textbf{clopen} if for each $V\in\VN$, the set $\ps S_V$ is a clopen subset of the Gel'fand spectrum $\Sig_V$. The set of all clopen subobjects of $\Sig$ is denoted as $\Subcl{\Sig}$.
\end{definition}
The set $\Subcl{\Sig}$, together with the lattice operations and bi-Heyting algebra structure defined below, is the algebraic implementation of the new topos-based form of quantum logic. The elements $\ps S\in\Subcl{\Sig}$ represent propositions about the values of the physical quantities of the quantum system. The most direct connection with propositions of the form ``$\Ain\De$'' is given by the map called daseinisation, see Def. \ref{Def_OuterDas} below.

We note that the concept of contextuality (cf. section \ref{Sec_Background}) is implemented by this construction, since $\Sig$ is a presheaf over the context category $\VN$. Moreover, coarse-graining is mathematically realised by the fact that we use subobjects of presheaves. In the case of $\Sig$ and its clopen subobjects, this means the following: for each context $V\in\VN$, the component $\ps S_V\subseteq\Sig_V$ represents a \emph{local proposition} about the value of some physical quantity. If $V'\subset V$, then $\ps S_{V'}\supseteq\Sig(i_{V'V})(\ps S_V)$ (since $\ps S$ is a subobject), so $\ps S_{V'}$ represents a local proposition at the smaller context $V'\subset V$ that is \emph{coarser} than (i.e., a consequence of) the local proposition represented by $\ps S_V$.

A clopen subobject $\ps S\in\Subcl{\Sig}$ can hence be interpreted as a collection of \emph{local propositions}, one for each context, such that smaller contexts are assigned coarser propositions.

Clearly, the definition of clopen subobjects makes use of the Gel'fand topologies on the components $\Sig_V$, $V\in\VN$. We note that for each abelian von Neumann algebra $V$ (and hence for each context $V\in\VN$), there is an isomorphism of complete Boolean algebras
\begin{align}			\label{Eq_alphaV}
			\alpha_V:\PV &\lra \Cp(\Sig_V)\\			\nonumber
			\hP &\lmt \{\ld\in\Sig_V \mid \ld(\hP)=1\}.
\end{align}
Here, $\Cp(\Sig_V)$ denotes the clopen subsets of $\Sig_V$.

There is a purely order-theoretic description of $\Subcl\Sig$: let 
\begin{equation}
			\mc P:=\prod_{V\in\VN}\PV
\end{equation}
be the set of choice functions $f:\VN\ra\coprod_{V\in\VN}\PV$, where $f(V)\in\PV$ for all $V\in\VN$. Equipped with pointwise operations, $\mc P$ is a complete Boolean algebra, since each $\PV$ is a complete Boolean algebra. Consider the subset $\mc S$ of $\mc P$ consisting of those functions for which $V'\subset V$ implies $f(V')\geq f(V)$. The subset $\mc S$ is closed under all meets and joins (in $\mc P$), and clearly, $\mc S\simeq\Subcl\Sig$.

We define a partial order on $\Subcl{\Sig}$ in the obvious way:
\begin{equation}
			\forall \ps S,\ps T\in\Subcl{\Sig} :\quad \ps S \leq \ps T :\Longleftrightarrow (\forall V\in\VN: \ps S_V\subseteq \ps T_V).
\end{equation}
We define the corresponding (complete) lattice operations in a stagewise manner, i.e., at each context $V\in\VN$ separately: for any family $(\ps S_i)_{i\in I}$,
\begin{equation}
			\forall V\in\VN : (\bmeet_{i\in I}\ps S_i)_V:=\on{int}(\bigcap_{i\in I}\ps S_{i;V}),
\end{equation}
where $\ps S_{i;V}\subseteq\Sig_V$ is the component at $V$ of the clopen subobject $\ps S_i$. Note that the lattice operation is not just componentwise set-theoretic intersection, but rather the interior (with respect to the Gel'fand topology) of the intersection. This guarantees that one obtains clopen subsets at each stage $V$, not just closed ones. Analogously,
\begin{equation}
			\forall V\in\VN : (\bjoin_{i\in I}\ps S_i)_V:=\on{cl}(\bigcup_{i\in I}\ps S_{i;V}),
\end{equation}
where the closure of the union is necessary in order to obtain clopen sets, not just open ones. The fact that meets and joins are not given by set-theoretic intersections and unions also means that $\Subcl{\Sig}$ is not a sub-Heyting algebra of the Heyting algebra $\Sub{\Sig}$ of all subobjects of the spectral presheaf. The difference between $\Sub{\Sig}$ and $\Subcl{\Sig}$ is analogous to the difference between the power set $PX$ of a set $X$ and the complete Boolean algebra $BX$ of measurable subsets (with respect to some measure) modulo null sets. For results on measures and quantum states from the perspective of the topos approach, see \cite{Doe08,DI12}.

In section \ref{Sec_RepOfPropsAndBiHeyting}, we will show that $\Subcl{\Sig}$ is a complete bi-Heyting algebra.

\begin{example}
For illustration, we consider a simple example: let $\N$ be the \emph{abelian} von Neumann of diagonal matrices in $3$ dimensions. This is given by
\begin{equation}
			\N:=\bbC\hP_1+\bbC\hP_2+\bbC\hP_3,
\end{equation}
where $\hP_1,\hP_2,\hP_3$ are pairwise orthogonal rank-$1$ projections on a $3$-dimensional Hilbert space. The projection lattice $\PN$ of $\N$ has $8$ elements,
\begin{equation}
			\PN=\{\hat 0,\hP_1,\hP_2,\hP_2,\hP_3,\hP_1+\hP_2,\hP_1+\hP_3,\hP_2+\hP_3,\hat 1\}.
\end{equation}
Of course, $\PN$ is a Boolean algebra.

The algebra $\N$ has three non-trivial abelian subalgebras,
\begin{equation}
			V_i:=\bbC\hP_i+\bbC(\hat 1-\hP_i),\quad i=1,2,3.
\end{equation}
Hence, the context category $\VN$ is the $4$-element poset with $\N$ as top element and $V_i\subset\N$ for $i=1,2,3$.

The Gel'fand spectrum $\Sig_{\N}$ of $\N$ has three elements $\ld_1,\ld_2,\ld_3$ such that
\begin{equation}
			\ld_i(\hP_j)=\delta_{ij}.
\end{equation}
The Gel'fand spectrum $\Sig_{V_1}$ of $V_1$ has two elements $\ld'_1,\ld'_{2+3}$ such that
\begin{equation}
			\ld'_1(\hP_1)=1,\quad\ld'_1(\hat 1-\hP_1)=0,\quad\ld'_{2+3}(\hP_1)=0,\quad\ld'_{2+3}(\hat 1-\hP_1)=1.
\end{equation}
(Note that $\hat 1-\hP_1=\hP_2+\hP_3$.) Analogously, the spectrum $\Sig_{V_2}$ has two elements $\ld'_{1+3},\ld'_2$, and the spectrum $\Sig_{V_3}$ has two elements $\ld'_{1+2},\ld'_3$.

Consider the restriction map of the spectral presheaf from $\Sig_{\N}$ to $\Sig_1$:
\begin{equation}
			\Sig(i_{V_1,\N})(\ld_1)=\ld'_1,\quad\Sig(i_{V_1,\N})(\ld_2)=\Sig(i_{V_1,\N})(\ld_3)=\ld'_{2+3}.
\end{equation}
The restriction maps from $\Sig_{\N}$ to $\Sig_{V_2}$ resp. $\Sig_{V_3}$ are defined analogously. This completes the description of the spectral presheaf $\Sig$ of the algebra $\N$.

We will now determine all clopen subobjects of $\Sig$. First, note that the Gel'fand spectra all are discrete sets, so topological questions are trivial here. We simply have to determine all subobjects of $\Sig$. We distinguish a number of cases:
\begin{itemize}
	\item [(a)] Let $\ps S\in\Subcl\Sig$ be a subobject such that $\ps S_{\N}=\Sig_{\N}=\{\ld_1,\ld_2,\ld_3\}$. Then the restriction maps of $\Sig$ dictate that for each $V_i$, $i=1,2,3$, we have $\ps S_{V_i}\supset\Sig(i_{V_i,\N})(\ps S_{N})=\Sig_{V_i}$, so $\ps S$ must be $\Sig$ itself.
	
	\item [(b)] Let $\ps S$ be a subobject such that $\ps S_{\N}$ contains two elements, e.g. $\ps S_{\N}=\{\ld_1,\ld_2\}$. Then $\ps S_{V_1}=\Sig_{V_1}$ and $\ps S_{V_2}=\Sig_{V_2}$, but $\ps S_{V_3}$ can either be $\{\ld'_{1+2}\}$ or $\{\ld'_{1+2},\ld'_3\}$, so there are $2$ options. Moreover, there are three ways of picking two elements from the three-element set $\Sig_{\N}$, so we have $3\cdot 2=6$ subobjects $\ps S$ with two elements in $\ps S_{\N}$.
	
	\item [(c)] Let $\ps S$ be such that $\ps S_{\N}$ contains one element, e.g. $\ps S_{\N}=\{\ld_1\}$. Then $\ps S_{V_1}$ can either be $\{\ld'_1\}$ or $\{\ld'_1,\ld'_{2+3}\}$; $\ps S_{V_2}$ can either be $\{\ld'_{1+3}\}$ or $\{\ld'_{1+3},\ld'_2\}$; and $\ps S_{V_3}$ can either be $\{\ld'_{1+2}\}$ or $\{\ld'_{1+2},\ld'_3\}$. Hence, there are $2^3$ options. Moreover, there are three ways of picking one element from $\Sig_{\N}$, so there are $3\cdot 2^3=24$ subobjects $\ps S$ with one element in $\ps S_{\N}$.
	
	\item [(d)] Finally, consider a subobject $\ps S$ such that $\ps S_{\N}=\emptyset$. Since the $V_i$ are not contained in one another, there are no conditions arising from restriction maps of the spectral presheaf $\Sig$. Hence, we can pick an arbitrary subset of $\Sig_{V_i}$ for $i=1,2,3$. Since each $\Sig_{V_i}$ has $2$ elements, there are $4$ subsets of each, so we have $4^3=64$ subobjects $\ps S$ with $\ps S_{\N}=\emptyset$.
\end{itemize}
In all, $\Subcl\Sig$ has $64+24+6+1=95$ elements.
\end{example}

We conclude this section with the remark that the pertinent topos in which the spectral presheaf (and the other presheaves discussed in this section) lie of course is the topos $\SetVNop$ of presheaves over the context category $\VN$.

%----------------------------------------------------
\section{Representation of propositions and bi-Heyting algebra structure}			\label{Sec_RepOfPropsAndBiHeyting}
\begin{definition}			\label{Def_OuterDas}
Let $\N$ be a von Neumann algebra, and let $\PN$ be its lattice of projections. The map
\begin{align}
			\ps\deo: \PN &\lra \Subcl{\Sig}\\			\nonumber
			\hP &\lmt \ps\deo(\hP):=(\alpha_V(\deo_{\N,V}(\hP)))_{V\in\VN}
\end{align}
is called \textbf{outer daseinisation of projections}.
\end{definition}
This map was introduced in \cite{DI(2)} and discussed in detail in \cite{Doe09,Doe11}. It can be seen as a `translation' map from standard quantum logic, encoded by the complete orthomodular lattice $\PN$ of projections, to a form of (super)intuitionistic logic for quantum systems, based on the clopen subobjects of the spectral presheaf $\Sig$, which conceptually plays the r\^ole of a quantum state space.

In standard quantum logic, the projections $\hP\in\PN$ represent propositions of the form ``$\Ain\De$'', that is, ``the physical quantity $A$ has a value in the Borel set $\De$ of real numbers''. The connection between propositions and projections is given by the spectral theorem. Outer daseinisation can hence be seen as a map from propositions of the form ``$\Ain\De$'' into the bi-Heyting algebra $\Subcl{\Sig}$ of clopen subobjects of the spectral presheaf. A projection $\hP$, representing a proposition ``$\Ain\De$'', is mapped to a collection $(\deo_{\N,V}(\hP))_{V\in\VN}$, consisting of one projection $\deo_{\N,V}(\hP)$ for each context $V\in\VN$. (Each isomorphism $\alpha_V$, $V\in\VN$, just maps the projection $\deo_{\N,V}(\hP)$ to the corresponding clopen subset of $\Sig_V$, which does not affect the interpretation.)

Since we have $\deo_{\N,V}(\hP)\geq\hP$ for all $V$, the projection $\deo_{\N,V}(\hP)$ represents a coarser (local) proposition than ``$\Ain\De$'' in general. For example, if $\hP$ represents ``$\Ain\De$'', then $\deo_{\N,V}(\hP)$ may represent ``$\Ain\Ga$'' where $\Ga\supset\De$. 

The map $\ps\deo$ preserves all joins, as shown in section 2.D of \cite{DI(2)} and in \cite{Doe09}. Here is a direct argument: being left adjoint to the inclusion of $\PV$ into $\PN$, the map $\deo_{\N,V}$ preserves all colimits, which are joins. Moreover, $\alpha_V$ is an isomorphism of complete Boolean algebras, so $\alpha_V\circ\deo_{\N,V}$ preserves all joins. This holds for all $V\in\VN$, and joins in $\Subcl{\Sig}$ are defined stagewise, so $\ps\deo$ preserves all joins. 

Moreover, $\ps\deo$ is order-preserving and injective, but not surjective. Clearly, $\ps\deo(\hat 0)=\ps 0$, the empty subobject, and $\ps\deo(\hat 1)=\Sig$. For meets, we have
\begin{equation}
			\forall \hP,\hQ\in\PN : \ps\deo(\hP\meet\hQ)\leq\ps\deo(\hP)\meet\ps\deo(\hQ).
\end{equation}
In general, $\ps\deo(\hP)\meet\ps\deo(\hQ)$ is not of the form $\ps\deo(\hat R)$ for any projection $\hat R\in\PN$. See \cite{DI(2),Doe09} for proof of these statements.

Let $(\ps S_i)_{i\in I}\subseteq\Subcl{\Sig}$ be a family of clopen subobjects of $\Sig$, and let $\ps S\in\Subcl{\Sig}$. Then
\begin{equation}
			\forall V\in\VN : (\ps S\meet\bjoin_{i\in I}\ps S_i)_V=\bjoin_{i\in I}(\ps S_V\meet\ps S_{i;V}),
\end{equation}
since $\Cp(\Sig_V)$ is a distributive lattice (in fact, a complete Boolean algebra) in which finite meets distribute over arbitrary joins. Hence, for each $\ps S\in\Subcl{\Sig}$, the functor
\begin{equation}
			\ps S\meet\_:\Subcl{\Sig} \lra \Subcl{\Sig}
\end{equation}
preserves all joins, so by the adjoint functor theorem for posets, it has a right adjoint
\begin{equation}
			\ps S\Rightarrow\_:\Subcl{\Sig} \lra \Subcl{\Sig}.
\end{equation}
This map, the \textbf{Heyting implication from $\ps S$}, makes $\Subcl{\Sig}$ into a complete Heyting algebra. This was shown before in \cite{DI(2)}. The Heyting implication is given by the adjunction
\begin{equation}
			 \ps R\meet\ps S\leq\ps T\quad\text{if and only if}\quad\ps R\leq (\ps S\Rightarrow\ps T).
\end{equation}
(Note that $\ps S\meet\_=\_\meet\ps S$.) This implies that
\begin{equation}
			(\ps S\Rightarrow\ps T)=\bjoin\{\ps R\in\Subcl{\Sig} \mid \ps R\meet\ps S\leq\ps T\}.
\end{equation}
The stagewise definition is: for all $V\in\VN$,
\begin{equation}
			(\ps S\Rightarrow\ps T)_V=\{\ld\in\Sig_V \mid \forall V'\subseteq V:\text{ if } \ld|_{V'}\in \ps S_{V'}\text{, then }\ld|_{V'}\in\ps T_{V'}\}.
\end{equation}
As usual, the \textbf{Heyting negation $\neg$} is defined for all $\ps S\in\Subcl{\Sig}$ by
\begin{equation}
			\neg\ps S:=(\ps S\Rightarrow\ps 0).
\end{equation}
That is, $\neg\ps S$ is the largest element of $\Subcl{\Sig}$ such that
\begin{equation}
			\ps S\meet\neg\ps S=\ps 0.
\end{equation}
The stagewise expression for $\neg\ps S$ is 
\begin{equation}			\label{Eq_HeytNegStagew}
			(\neg\ps S)_V=\{\ld\in\Sig_V \mid \forall V'\subseteq V:\ld|_{V'}\notin \ps S_{V'}\}.
\end{equation}
In $\Subcl{\Sig}$, we also have, for all families $(\ps S_i)_{i\in I}\subseteq\Subcl{\Sig}$ and all $\ps S\in\Subcl{\Sig}$,
\begin{equation}
				\forall V\in\VN: (\ps S\join\bmeet_{i\in I}\ps S_i)_V=\bmeet_{i\in I}(\ps S_V\join\ps S_{i;V}),
\end{equation}
since finite joins distribute over arbitrary meets in $\Cp(\Sig)$. Hence, for each $\ps S$ the functor
\begin{equation}
			\ps S\join\_:\Subcl{\Sig} \lra \Subcl{\Sig}
\end{equation}
preserves all meets, so it has a left adjoint
\begin{equation}
			\ps S\Leftarrow\_: \Subcl{\Sig} \lra \Subcl{\Sig}
\end{equation}
which we call \textbf{co-Heyting implication}. This map makes $\Subcl{\Sig}$ into a complete co-Heyting algebra. It is characterised by the adjunction
\begin{equation}
			(\ps S\Leftarrow\ps T)\leq\ps R\quad\text{iff}\quad\ps S\leq\ps T\join\ps R,
\end{equation}
so
\begin{equation}			\label{Eq_coImp}
			(\ps S\Leftarrow\ps T)=\bmeet\{\ps R\in\Subcl{\Sig} \mid \ps S\leq\ps T\join\ps R\}.
\end{equation}
One can think of $\ps S\Leftarrow\_$ as a kind of `subtraction' (see e.g. \cite{RZ96}): $\ps S\Leftarrow\ps T$ is the smallest clopen subobject $\ps R$ for which $\ps T\join\ps R$ is bigger then $\ps S$, so it encodes how much is `missing' from $\ps T$ to cover $\ps S$.

We define a \textbf{co-Heyting negation} for each $\ps S\in\Subcl{\Sig}$ by
\begin{equation}
			\sim\ps S:=(\Sig\Leftarrow\ps S).
\end{equation}
(Note that $\Sig$ is the top element in $\Subcl{\Sig}$.) Hence, $\sim\ps S$ is the smallest clopen subobject such that
\begin{equation}
			\sim\ps S\join\ps S=\Sig
\end{equation}
holds. We have shown in a direct manner, without use of topos theory as in section \ref{Sec_SigEtc}:

\begin{proposition}
$(\Subcl{\Sig},\meet,\join,\ps 0,\Sig,\Rightarrow,\neg,\Leftarrow,\sim)$ is a complete bi-Heyting algebra.
\end{proposition}

We give direct arguments for the following two facts (which also follow from the general theory of bi-Heyting algebras):
\begin{lemma}
For all $\ps S\in\Subcl{\Sig}$, we have $\neg\ps S\leq\;\sim\ps S$.
\end{lemma}

\begin{proof}
For all $V\in\VN$, it holds that $(\neg\ps S)_V\subseteq\Sig_V\backslash\ps S_V$, since $(\neg\ps S\meet\ps S)_V=(\neg\ps S)_V\cap\ps S_V=\emptyset$, while $(\sim\ps S)_V\supseteq\Sig_V\backslash\ps S_V$ since $(\sim\ps S\join\ps S)_V=(\sim\ps S)_V\cup\ps S_V=\Sig_V$.
\end{proof}

The above lemma and the fact that $\neg\ps S$ is the largest subobject such that $\neg\ps S\meet\ps S=\ps 0$ imply
\begin{corollary}
In general, $\sim\ps S\meet\ps S\geq\ps 0$.
\end{corollary}
This means that the co-Heyting negation does not give a system in which a central axiom of most logical systems, viz. freedom from contradiction, holds. We have a glimpse of \emph{paraconsistent logic}.

In fact, a somewhat stronger result holds: for any von Neumann algebra except for $\bbC\hat 1=M_1(\bbC)$ and $M_2(\bbC)$, we have $\sim\ps S>\neg\ps S$ and $\sim\ps S\meet\ps S>\ps 0$ for all clopen subobjects except $\ps 0$ and $\Sig$. This follows easily from the representation of clopen subobjects as families of projections, see beginning of next section.

%---------------------------------------------------
\section{Negations and regular elements}			\label{Sec_NegsAndRegs}
In this section, we will examine the Heyting negation $\neg$ and the co-Heyting negation $\sim$ more closely. We will determine regular elements with respect to the Heyting and the co-Heyting algebra structure.

Throughout, we will make use of the isomorphism $\alpha_V:\PV\ra\Cp(\Sig_V)$ (defined in (\ref{Eq_alphaV})) between the complete Boolean algebras of projections in an abelian von Neumann algebra $V$ and the clopen subsets of its spectrum $\Sig_V$. Given a projection $\hP\in\PV$, we will use the notation $S_{\hP}:=\alpha_V(\hP)$. Conversely, for $S\in\Cp(\Sig_V)$, we write $\hP_S:=\alpha_V^{-1}(S)$.

Given a clopen subobject $\ps S\in\Subcl{\Sig}$, it is useful to think of it as a collection of projections: consider
\begin{equation}
	(\hP_{\ps S_V})_{V\in\VN} = (\alpha_V(\ps S_V))_{V\in\VN},
\end{equation}
which consists of one projection for each context $V$. The fact that $\ps S$ is a subobject then translates to the fact that if $V'\subset V$, then $\hP_{\ps S_{V'}}\geq\hP_{\ps S_V}$. (This is another instance of coarse-graining.)

If $\ld\in\Sig_V$ and $\hP\in\PV$, then
\begin{equation}
			\ld(\hP)=\ld(\hP^2)=\ld(\hP)^2\in\{0,1\},
\end{equation}
where we used that $\hP$ is idempotent and that $\ld$ is multiplicative. 

\textbf{Heyting negation and Heyting-regular elements.} We consider the stagewise expression (see eq. (\ref{Eq_HeytNegStagew})) for the Heyting negation:
\begin{align}
		(\neg\ps S)_V &=\{\ld\in\Sig_V \mid \forall V'\subseteq V:\ld|_{V'}\notin \ps S_{V'}\}\\
		&= \{\ld\in\Sig_V \mid \forall V'\subseteq V:\ld|_{V'}(\hP_{\ps S_{V'}})=0 \}\\
		&= \{\ld\in\Sig_V \mid \forall V'\subseteq V:\ld(\hP_{\ps S_{V'}})=0 \}\\
		&= \{\ld\in\Sig_V \mid \ld(\bjoin_{V'\subseteq V}\hP_{\ps S_{V'}})=0 \}
\end{align}
As we saw above, the smaller the context $V'$, the larger the associated projection $\hP_{\ps S_{V'}}$. Hence, for the join in the above expression, only the \emph{minimal} contexts $V'$ contained in $V$ are relevant. A minimal context is generated by a single projection $\hQ$ and the identity,
\begin{equation}
			V_{\hQ}:=\{\hQ,\hat 1\}''=\bbC\hQ+\bbC\hat 1.
\end{equation}
Here, it becomes important that we excluded the trivial context $V_0=\{\hat 1\}''=\bbC\hat 1$. Let
\begin{equation}
			m_V:=\{V'\subseteq V \mid V'\text{ minimal}\}=\{V_{\hQ} \mid \hQ\in\PV\}.
\end{equation}
We obtain
\begin{align}
			(\neg\ps S)_V &= \{\ld\in\Sig_V \mid \ld(\bjoin_{V'\in m_V}\hP_{\ps S_{V'}})=0 \}\\
			&= \{\ld\in\Sig_V \mid \ld(\hat 1-\bjoin_{V'\in m_V}\hP_{\ps S_{V'}})=1 \}\\
%			&= \{\ld\in\Sig_V \mid \ld(\bmeet_{V'\in m_V}(\hat 1-\hP_{\ps S_{V'}}))=1 \}\\
			&= S_{\hat 1-\bjoin_{V'\in m_V}\hP_{\ps S_{V'}}}.
\end{align}
This shows:
\begin{proposition}			\label{Prop_HeytingNegLocal}
Let $\ps S\in\Subcl{\Sig}$, and let $V\in\VN$. Then
\begin{equation}
			\hP_{(\neg\ps S)_V}=\hat 1-\bjoin_{V'\in m_V}\hP_{\ps S_{V'}},
\end{equation}
where $m_V=\{V'\subseteq V \mid V'\text{ minimal}\}$.
\end{proposition}

We can now consider double negation: $(\neg\neg\ps S)_V=S_{\hat 1-\bjoin_{V'\in m_V}\hP_{(\neg\ps S)_{V'}}}$, so
\begin{equation}
			\hP_{(\neg\neg\ps S)_V}=\hat 1-\bjoin_{V'\in m_V}\hP_{(\neg\ps S)_{V'}}.
\end{equation}
For a $V'\in m_V$, we have $\hP_{(\neg\ps S)_{V'}}=\hat 1-\bjoin_{W\in m_{V'}}\hP_{\ps S_{W}}$, but $m_{V'}=\{V'\}$, since $V'$ is minimal, so $\hP_{(\neg\ps S)_{V'}}=\hat 1-\hP_{\ps S_{V'}}$. Thus,
\begin{equation}			\label{Eq_DoubleHeytNegStagew}
			\hP_{(\neg\neg\ps S)_V}=\hat 1-\bjoin_{V'\in m_V}(\hat 1-\hP_{\ps S_{V'}})=\bmeet_{V'\in m_V}\hP_{\ps S_{V'}}.
\end{equation}
Since $\hP_{\ps S_{V'}}\geq\hP_{\ps S_V}$ for all $V'\in m_V$ (because $\ps S$ is a subobject), we have
\begin{equation}			\label{Eq_DoubleHeytNegBigger}
			\hP_{(\neg\neg\ps S)_V}=\bmeet_{V'\in m_V}\hP_{\ps S_{V'}}\geq\hP_{\ps S_V}
\end{equation}
for all $V\in\VN$, so $\neg\neg\ps S\geq\ps S$ as expected. We have shown:
\begin{proposition}			\label{Prop_HeytingReg}
An element $\ps S$ of $\Subcl{\Sig}$ is Heyting-regular, i.e., $\neg\neg\ps S=\ps S$, if and only if for all $V\in\VN$, it holds that
\begin{equation}
			\hP_{\ps S_V}=\bmeet_{V'\in m_V}\hP_{\ps S_{V'}},
\end{equation}
where $m_V=\{V'\subseteq V \mid V'\text{ minimal}\}$.
\end{proposition}

\begin{definition}
A clopen subobject $\ps S\in\Subcl{\Sig}$ is called \textbf{tight} if 
\begin{equation}
			\Sig(i_{V'V})(\ps S_V)=\ps S_{V'}
\end{equation}
for all $V',V\in\VN$ such that $V'\subseteq V$.
\end{definition}
For arbitrary subobjects, we only have $\Sig(i_{V'V})(\ps S_V)\subseteq\ps S_{V'}$. Let $\ps S\in\Subcl{\Sig}$ be an arbitrary clopen subobject, and let $V,V'\in\VN$ such that $V'\subset V$. Then
$\Sig(i_{V'V})(\ps S_V)\subseteq\ps S_{V'}\subseteq\Sig_{V'}$, so $\hP_{\Sig(i_{V'V})(\ps S_V)}\in\mc P(V')$. Thm. 3.1 in \cite{DI(2)} shows that
\begin{equation}
				\hP_{\Sig(i_{V'V})(\ps S_V)}=\deo_{V,V'}(\hP_{\ps S_V}).
\end{equation}
This key formula relates the restriction maps $\Sig(i_{V'V}):\Sig_V\ra\Sig_{V'}$ of the spectral presheaf to the maps $\deo_{V,V'}:\PV\ra\mc P(V')$. Using this, we see that
\begin{proposition}
A clopen subobject $\ps S\in\Subcl{\Sig}$ is tight if and only if $\hP_{\ps S_{V'}}=\deo_{V,V'}(\hP_{\ps S_V})$ for all $V',V\in\VN$ such that $V'\subseteq V$.
\end{proposition}
It is clear that all clopen subobjects of the form $\ps\deo(\hP)$, $\hP\in\PN$, are tight (see Def. \ref{Def_OuterDas}).
\begin{proposition}
For a tight subobject $\ps S\in\Subcl{\Sig}$, it holds that $\neg\neg\ps S=\ps S$, i.e., tight subobjects are Heyting-regular.
\end{proposition}

\begin{proof}
We saw in equation (\ref{Eq_DoubleHeytNegStagew}) that $\hP_{(\neg\neg\ps S)_V}=\bmeet_{V'\in m_V}\hP_{\ps S_{V'}}$ for all $V\in\VN$. Moreover, $\hP_{(\neg\neg\ps S)_V}\geq\hP_{\ps S_V}$ from equation (\ref{Eq_DoubleHeytNegBigger}). Consider the minimal subalgebra $V_{\hP_{\ps S_V}}=\{\hP_{\ps S_V},\hat 1\}''$ of $V$. Then, since $\ps S$ is tight, we have
\begin{equation}
			\deo_{V,V_{\hP_{\ps S_V}}}(\hP_{\ps S_V})=\bmeet\{\hQ\in\mc P(V_{\hP_{\ps S_V}}) \mid \hQ\geq\hP_{\ps S_V}\}=\hP_{\ps S_V},
\end{equation}
so, for all $V\in\VN$,
\begin{equation}
			\hP_{(\neg\neg\ps S)_V}=\bmeet_{V'\in m_V}\hP_{\ps S_{V'}}=\hP_{\ps S_V}.
\end{equation} 
\end{proof}

\begin{corollary}
Outer daseinisation $\ps\deo:\PN\ra\Subcl{\Sig}$ maps projections into the Heyting-regular elements of $\Subcl{\Sig}$.
\end{corollary}

We remark that in order to be Heyting-regular, an element $\ps S\in\Subcl{\Sig}$ need not be tight.

\textbf{Co-Heyting negation and co-Heyting regular elements.} For any $\ps S\in\Subcl{\Sig}$, by its defining property $\sim\ps S$ is the smallest element of $\Subcl{\Sig}$ such that $\ps S\join\sim\ps S=\Sig$.

Let $V$ be a maximal context, i.e., a maximal abelian subalgebra (masa) of the non-abelian von Neumann algebra $\N$. Then clearly
\begin{equation}
			(\sim\ps S)_V=\Sig_V\backslash\ps S_V.
\end{equation}
Let $V\in\VN$, not necessarily maximal. We define
\begin{equation}
			M_V:=\{\tilde V\supseteq V \mid \tilde V\text{ maximal}\}.
\end{equation}
\begin{proposition}			\label{Prop_CoHeytingNegLocal}
Let $\ps S\in\Subcl{\Sig}$, and let $V\in\VN$. Then
\begin{equation}
			\hP_{(\sim\ps S)_V}=\bjoin_{\tilde V\in M_V}(\deo_{\tilde V,V}(\hat 1-\hat P_{\ps S_{\tilde V}})),
\end{equation}
where $M_V=\{\tilde V\supseteq V \mid \tilde V\text{ maximal}\}$.
\end{proposition}

\begin{proof}
$\sim\ps S$ is a (clopen) subobject, so we must have
\begin{equation}
			\hP_{(\sim\ps S)_V}\geq\bjoin_{\tilde V\in M_V}(\deo_{\tilde V,V}(\hat 1-\hat P_{\ps S_{\tilde V}})),
\end{equation}
since $(\sim\ps S)_V$, the component at $V$, must contain all the restrictions of the components $(\sim\ps S)_{\tilde V}$ for $\tilde V\in M_V$ (and the above inequality expresses this using the corresponding projections).

On the other hand, $\sim\ps S$ is the \emph{smallest} clopen subobject such that $\newblock{\ps S\join\sim\ps S}=\Sig$. So it suffices to show that for 
$\hP_{(\sim\ps S)_V}=\bjoin_{\tilde V\in M_V}(\deo_{\tilde V,V}(\hat 1-\hP_{\ps S_{\tilde V}}))$, we have $\hP_{(\sim\ps S)_V}\join\hP_{\ps S_V}=\hat 1$ for all $V\in\VN$, and hence $\sim\ps S\join\ps S=\Sig$.

If $V$ is maximal, then $\hP_{(\sim\ps S)_V}=\deo_{V,V}(\hat 1-\hP_{\ps S_V})=\hat 1-\hP_{\ps S_V}$ and hence $\hP_{(\sim\ps S)_V}\join\hP_{\ps S_V}=\hat 1$. If $V$ is non-maximal and $\tilde V$ is any maximal context containing $V$, then $\hP_{(\sim\ps S)_V}\geq\hP_{(\sim\ps S)_{\tilde V}}$ and $\hP_{\ps S_V}\geq\hP_{\ps S_{\tilde V}}$, so $\hP_{(\sim\ps S)_V}\join\hP_{\ps S_V}\geq\hP_{(\sim\ps S)_{\tilde V}}\join\hP_{\ps S_{\tilde V}}=\hat 1$.
\end{proof}

For the double co-Heyting negation, we obtain
\begin{align}
			\hP_{(\sim\sim\ps S)_V} &=\bjoin_{\tilde V\in M_V}\deo_{\tilde V,V}(\hat 1-\hP_{(\sim\ps S)_{\tilde V}})\\
			&= \bjoin_{\tilde V\in M_V}\deo_{\tilde V,V}(\hat 1-\bjoin_{W\in M_{\tilde V}}\deo_{W,\tilde V}(\hat 1-\hP_{\ps S_W})).
\end{align}
Since $\tilde V$ is maximal, we have $M_{\tilde V}=\{\tilde V\}$, and the above expression simplifies to
\begin{align}
			\hP_{(\sim\sim\ps S)_V} &=\bjoin_{\tilde V\in M_V}\deo_{\tilde V,V}(\hat 1-(\hat 1-\hP_{\ps S_{\tilde V}}))\\
			&= \bjoin_{\tilde V\in M_{V}}\deo_{\tilde V,V}(\hP_{\ps S_{\tilde V}}).
\end{align}

Note that the fact that $\ps S$ is a subobject implies that 
\begin{equation}			\label{Eq_DoubleCoHeytNegBigger}
			\hP_{(\sim\sim\ps S)_V}\leq\hP_{\ps S_V}
\end{equation}
for all $V\in\VN$, so $\sim\sim\ps S\leq\ps S$ as expected. We have shown:

\begin{proposition}			\label{Prop_CoHeytReg}
An element $\ps S$ of $\Subcl{\Sig}$ is co-Heyting-regular, i.e., $\newblock{\sim\sim\ps S}=\ps S$, if and only if for all $V\in\VN$ it holds that
\begin{equation}
			\hP_{\ps S_V}=\bjoin_{\tilde V\in M_{V}}\deo_{\tilde V,V}(\hP_{\ps S_{\tilde V}}),
\end{equation}
where $M_V=\{\tilde V\supseteq V \mid \tilde V\text{ maximal}\}$.
\end{proposition}

\begin{proposition}
If $\ps S\in\Subcl{\Sig}$ is tight, then $\sim\sim\ps S=\ps S$, i.e., tight subobjects are co-Heyting regular.
\end{proposition}

\begin{proof}
If $\ps S$ is tight, then for all $V\in\VN$ and for all $\tilde V\in M_V$, we have $\hP_{\ps S_V}=\deo_{\tilde V,V}(\hP_{\ps S_{\tilde V}})$, so $\bjoin_{\tilde V\in M_V}\deo_{\tilde V,V}(\hP_{\ps S_{\tilde V}})=\hP_{\ps S_V}$. By Prop. \ref{Prop_CoHeytReg}, the result follows.
\end{proof}

\begin{corollary}
Outer daseinisation $\ps\deo:\PN\ra\Subcl{\Sig}$ maps projections into the co-Heyting-regular elements of $\Subcl{\Sig}$.
\end{corollary}

\textbf{Physical interpretation.} We conclude this section by giving a tentative physical interpretation of the two kinds of negation. For this interpretation, it is important to think of an element $\ps S\in\Subcl{\Sig}$ as a collection of local propositions $\ps S_V$ (resp. $\hP_{\ps S_V}$), one for each context $V$. Moreover, if $V'\subset V$, then the local proposition represented by $\ps S_{V'}$ is coarser than the local proposition represented by $\ps S_V$.

Let $\ps S\in\Subcl{\Sig}$ be a clopen subobject, and let $\neg\ps S$ be its Heyting complement. As shown in Prop. \ref{Prop_HeytingNegLocal}, the local expression for components of $\neg\ps S$ is given by
\begin{equation}
			\hP_{(\neg\ps S)_V}=\hat 1-\bjoin_{V'\in m_V}\hP_{\ps S_{V'}},
\end{equation}
where $m_V$ is the set of all minimal contexts contained in $V$. The projection $\hP_{(\neg\ps S)_V}$ is always smaller than or equal to $\hat 1-\hP_{\ps S_V}$, since $\hP_{\ps S_{V'}}\geq\hP_{\ps S_V}$ for all $V'\in m_V$. For the Heyting negation of the local proposition in the context $V$, represented by $\ps S_V$ or equivalently by the projection $\hP_{\ps S_V}$, one has to consider all the coarse-grainings of this proposition to minimal contexts (which are the `maximal' coarse-grainings). The Heyting complement $\neg\ps S$ is determined at each stage $V$ as the complement of the join of all the coarse-grainings $\hP_{\ps S_{V'}}$ of $\hP_{\ps S_V}$.

In other words, the component of the Heyting complement $\neg\ps S$ at $V$ is not simply the complement of $\ps S_V$, but the complement of the disjunction of all the coarse-grainings of this local proposition to all smaller contexts. The coarse-grainings of $\ps S_V$ are specified by the clopen subobject $\ps S$ itself.

The component of the co-Heyting complement $\sim\ps S$ at a context $V$ is given by
\begin{equation}
						\hP_{(\sim\ps S)_V}=\bjoin_{\tilde V\in M_V}(\deo_{\tilde V,V}(\hat 1-\hat P_{\ps S_{\tilde V}})),
\end{equation}
where $M_V$ is the set of maximal contexts containing $V$. The projection $\hP_{(\sim\ps S)_V}$ is always larger than or equal to $\hat 1-\hP_{\ps S_V}$, as was argued in the proof of Prop. \ref{Prop_CoHeytingNegLocal}. This means that the co-Heyting complement $\sim\ps S$ has a component $(\sim\ps S)_V$ at $V$ that may overlap with the component $\ps S_V$, hence the corresponding local propositions are not mutually exclusive in general. Instead, $\hP_{(\sim\ps S)_V}$ is the disjunction of all the coarse-grainings of complements of (finer, i.e., stronger) local propositions at contexts $\tilde V\supset V$.

The co-Heyting negation hence gives local propositions that for each context $V$ take into account all those contexts $\tilde V$ from which one can coarse-grain to $V$. The component $(\sim\ps S)_V$ is defined in such a way that all the stronger local propositions at maximal contexts $\tilde V\supset V$ are complemented in the usual sense, i.e., $\hP_{(\sim\ps S)_{\tilde V}}=\hat 1-\hP_{\ps S_{\tilde V}}$ for all maximal contexts $\tilde V$. At smaller contexts $V$, we have some coarse-grained local proposition, represented by $\hP_{(\sim\ps S)_V}$, that will in general not be disjoint from (i.e., mutually exclusive with) the local proposition represented by $\hP_{\ps S_V}$.

%-------------------------------------------------
\section{Conclusion and outlook}			\label{Sec_Conclusion}
Summing up, we have shown that to each quantum system described by a von Neumann algebra $\N$ of physical quantities one can associate a (generalised) quantum state space, the spectral presheaf $\Sig$, together with a complete bi-Heyting algebra $\Subcl{\Sig}$ of clopen subobjects. Elements $\ps S$ can be interpreted as families of local propositions, where `local' refers to contextuality: each component $\ps S_V$ of a clopen subobject represents a proposition about the value of a physical quantity in the context (i.e., abelian von Neumann subalgebra) $V$ of $\N$. Since $\ps S$ is a subobject, there is a built-in form of coarse-graining which guarantees that if $V'\subset V$ is a smaller context, then the local proposition represented by $\ps S_{V'}$ is coarser than the proposition represented by $\ps S_V$.

The map called outer daseinisation of projections (see Def. \ref{Def_OuterDas}) is a convenient bridge between the usual Hilbert space formalism and the new topos-based form of quantum logic. Daseinisation maps a propositions of the form ``$\Ain\De$'', represented by a projection $\hP$ in the complete orthomodular lattice $\PN$ of projections in the von Neumann algebra $\N$, to an element $\ps\deo(\hP)$ of the bi-Heyting algebra $\Subcl{\Sig}$. 

We characterised the two forms of negation arising from the Heyting and the co-Heyting structure on $\Subcl{\Sig}$ by giving concrete stagewise expressions (see Props. \ref{Prop_HeytingNegLocal} and \ref{Prop_CoHeytingNegLocal}), considered double negation and characterised Heyting regular elements of $\Subcl{\Sig}$ (Prop. \ref{Prop_HeytingReg}) as well as co-Heyting regular elements (Prop. \ref{Prop_CoHeytReg}). It turns out that daseinisation maps projections into Heyting regular and co-Heyting regular elements of the bi-Heyting algebra of clopen subobjects.

The main thrust of this article is to replace the standard algebraic representation of quantum logic in projection lattices of von Neumann algebras by a better behaved form based on bi-Heyting algebras. Instead of having a non-distributive orthomodular lattice of projections, which comes with a host of well-known conceptual and interpretational problems, one can consider a complete bi-Heyting algebra of propositions. In particular, this provides a distributive form of quantum logic. Roughly speaking, a non-distributive lattice with an orthocomplement has been traded for a distributive one with two different negations.

We conclude by giving some open problems for further study:
\begin{itemize}
	\item [(a)] It will be interesting to see how far the constructions presented in this article can be generalised beyond the case of von Neumann algebras. A generalisation to complete orthomodular lattices is immediate, but more general structures used in the study of quantum logic(s) remain to be considered. 
	\item [(b)] Bi-Heyting algebras are related to bitopological spaces, see \cite{BBGK10} and references therein. But the spectral presheaf $\Sig$ is not a topological (or bitopological) space in the usual sense. Rather, it is a presheaf which has no global elements. Hence, there is no direct notion of points available, which makes it impossible to define of a set underlying the topology (or topologies). Generalised notions of topology like frames will be useful to study the connections with bitopological spaces. 
	\item [(c)] All the arguments given in this article are topos-external. There is an internal analogue of the bi-Heyting algebra $\Subcl{\Sig}$ in the form of the power object $P\ps O$ of the so-called outer presheaf, see \cite{DI12}, so one can study many aspects internally in the topos $\SetVNop$ associated with the quantum system. This also provides the means to go beyond propositional logic to predicate logic, since each topos possesses an internal higher-order intuitionistic logic.
\end{itemize}

\vspace{0.5cm}

\textbf{Acknowledgements.} I am very grateful to the ASL, and to Reed Solomon, Valentina Harizanov and Jennifer Chubb personally, for giving me the opportunity to organise a Special Session on ``Logic and Foundations of Physics'' for the 2010 North American Meeting of the ASL, Washington D.C., March 17--20, 2010. I would like to thank Chris Isham and Rui Soares Barbosa for discussions and support. Many thanks to Dan Marsden, who read the manuscript at an early stage and made some valuable comments and suggestions. The anonymous referee also provided some very useful suggestions, which I incorporated. Finally, Dominique Lambert's recent talk at \textsl{Categories and Physics 2011} at Paris 7 served as an eye-opener on paraconsistent logic (and made me lose my fear of contradictions ;-) ).

\end{document}